\documentclass[12pt,draft,titlepage]{article}

\usepackage[english]{babel}
\usepackage{amsthm}
\usepackage{amssymb}
\usepackage{latexsym}
\usepackage[round]{natbib}
\usepackage{url}
\usepackage{bm}

\usepackage[final]{graphicx}

\usepackage{layout}

\usepackage{color}

\newcommand\ind{\bot\!\!\!\bot}
\newcommand{\g}{\,\vert\,}

\newcommand{\bc}{\begin{center}}
\newcommand{\ec}{\end{center}}

\newcommand{\bit}{\begin{itemize}}
\newcommand{\eit}{\end{itemize}}

\newcommand{\be}{\begin{eqnarray*}}
\newcommand{\ee}{\end{eqnarray*}}
\newcommand{\ben}{\begin{eqnarray}}
\newcommand{\een}{\end{eqnarray}}

\newcommand\M{\mathcal{M}}

\newcommand{\vat}{\mathbb{E}}
\newcommand{\var}{\mathbb{V}\!\mathrm{ar}}

\newcommand{\cov}{\mathbb{C}\mathrm{ov}}

\newcommand{\tr}{\mathrm{tr}}

\newcommand{\G}{\mathcal{G}}
\newcommand{\C}{\mathcal{C}}
\newcommand{\D}{\mathcal{D}}
\newcommand{\pa}{\mathrm{pa}}
\newcommand{\fa}{\mathrm{fa}}

\newcommand{\N}{\mathcal{N}}
\newcommand{\W}{\mathcal{W}}

\theoremstyle{plain}

\newtheorem{prop}{Proposition}[section]

\newcommand{\bone}{\bm{1}}
\newcommand{\bzero}{\bm{0}}

\newcommand{\bB}{\bm{B}}
\newcommand{\bC}{\bm{C}}
\newcommand{\bD}{\bm{D}}
\newcommand{\bE}{\bm{E}}
\newcommand{\bI}{\bm{I}}
\newcommand{\bK}{\bm{K}}
\newcommand{\bM}{\bm{M}}
\newcommand{\bR}{\bm{R}}

\newcommand{\bX}{\bm{X}}
\newcommand{\bY}{\bm{Y}}
\newcommand{\bZ}{\bm{Z}}

\newcommand{\bem}{\bm{m}}
\newcommand{\bu}{\bm{u}}
\newcommand{\bv}{\bm{v}}
\newcommand{\bx}{\bm{x}}
\newcommand{\by}{\bm{y}}

\newcommand{\bPhi}{\bm{\Phi}}
\newcommand{\bSigma}{\bm{\Sigma}}
\newcommand{\bPsi}{\bm{\Psi}}
\newcommand{\bOmega}{\bm{\Omega}}

\newcommand{\balpha}{\bm{\alpha}}

\newcommand{\bgamma}{\bm{\gamma}}
\newcommand{\btheta}{\bm{\theta}}
\newcommand{\bmu}{\bm{\mu}}

\begin{document}

\setcounter{footnote}{3}

\title{Objective Bayes Covariate-Adjusted\\
Sparse Graphical Model Selection\footnote{\textit{Running headline}:
Covariate-Adjusted Sparse DAG Selection}}


\author{Guido Consonni\\
Universit\`{a} Cattolica del Sacro Cuore\\
\url{guido.consonni@unicatt.it}
\and Luca La Rocca
\footnote{Corresponding author}\\
Universit\`a di Modena e Reggio Emilia\\
\url{luca.larocca@unimore.it}}

\maketitle


\begin{abstract} 
We present an objective Bayes method for covariance selection
in Gaussian multivariate regression models whose error term has a covariance structure
which is Markov with respect to a Directed Acyclic Graph (DAG).
The scope is covariate-adjusted sparse graphical model selection,
a topic of growing importance especially in the area of
genetical genomics (eQTL analysis).
Specifically, we provide a closed-form expression for
the marginal likelihood of any DAG (with small parent sets)
whose computation virtually requires no subjective elicitation by the user
and involves only conjugate matrix normal Wishart distributions.
This is made possible by a specific form of prior assignment,
whereby only one prior under the complete DAG model need be specified,
based on the notion of fractional Bayes factor.
All priors under the other DAG models are derived using prior modularity,
and global parameter independence, in the terminology of \citet{Geig:Heck:2002}.  
Since the marginal likelihood we obtain is constant within each class
of Markov equivalent DAGs, our method naturally specializes to
covariate-adjusted decomposable graphical models.
\medskip

\end{abstract}

\section{Introduction}
\label{sec:intro}

Graphical models are a well-established tool in multivariate statistics.
They allow to simplify high-dimensional distributions,
both in terms of computations and in terms of interpretation,
on the basis of a graph representing independencies between variables.
We assume the reader is familiar with the basic theory of undirected
and acyclic directed graphical models,
as presented for instance in \cite{Cowe:Dawi:Laur:Spie:1999},
or \cite{Laur:1996}; see also \cite{Whit:1990}.

Our interest lies in a collection of $q$ random variables
whose joint distribution, having density with respect to a product measure,
embodies a conditional independence structure
which can be represented by a Directed Acyclic Graph (DAG).
This means that each variable is conditionally independent of
its non-descendants given its parents; see Cowell et al. (1999, sect.~5.3).
Such a distribution is said to be Markov with respect to the DAG.
A~DAG model is a (parametric) family of multivariate distributions
which are Markov with respect to a DAG.
We will consider in particular Gaussian DAG models.
Then, the DAG structure will be reflected in the covariance matrix $\bSigma$:
if the DAG is complete, $\bSigma$ will be unconstrained;
for an incomplete DAG, $\bSigma$ will present constrained entries.
Notice that an unconstrained covariance matrix still has to be
s.p.d.~(symmetric positive definite).

Typically, the DAG structure is unknown,
and we want to infer it from $n$ joint observations of the $q$ variables.
From a Bayesian viewpoint one starts with a prior distribution
on the collection of all DAGs (prior on model space),
as well as with a prior distribution on the parameter space
of each given DAG (parameter prior).
Given these prior inputs, Bayesian inference produces a posterior probability
on the space of all DAGs, which summarizes all the uncertainty
in the light of the available data.
Several papers have addressed this problem for the case in which
the $n$ observations are i.i.d.~(independent and identically distributed)
conditionally on the parameters of the model;
see for instance \citet{Dawi:Laur:1993, Spieetal:1993, Hecketal:1995, Madi:Etal:1996}.
Of special interest for this paper is the work by \citet{Geig:Heck:2002};
see also \citet{Cons:Laro:2012} and \citet{Kuip:Moff:Heck:2014} for corrections.
\citet{Geig:Heck:2002} listed a set of assumptions on the collection of parameter priors
(across DAGs) which permit their construction starting from a single parameter prior
under a complete DAG (a DAG with all pairs of vertices directly connected).
This represents a dramatic simplification because:
i)~the specification of only one distribution is required,
while all the remaining priors are derived from this one;
ii)~the latter distribution is placed on an unconstrained parameter space
describing the model with no independencies.
In the Gaussian case ii) means that one can use a standard Inverse Wishart
on the covariance matrix, equivalently a Wishart on the corresponding precision matrix
(defined as the inverse of the covariance matrix)
so that the marginal likelihood can be expressed in closed form.


Different DAGs may define the same DAG model,
in which case they are called Markov equivalent.
Accordingly, the set of all DAGs for the $q$ variables
can be partitioned into Markov equivalence classes
(corresponding to distinct DAG models).
If DAGs are meant to specify exclusively conditional independencies,
as opposed to causal relationships \citep {Laur:2001, Dawi:2003},
then all DAGs specifying the same set of conditional independencies
should be regarded as indistinguishable using observational data.
The  method by \citet{Geig:Heck:2002} ensures that DAGs belonging to
the same equivalence class obtain the same marginal likelihood.
As a consequence, their method can also be used to infer decomposable graph structures,
by simply replacing each structure with an equivalent DAG (no matter which).

Despite its many advantages, the inferential procedure proposed by \citet{Geig:Heck:2002}
still  requires the  specification of a  potentially high-dimensional parameter prior
(especially  in large $q$ settings).
This naturally suggests an objective Bayes approach,
which is virtually free from prior elicitation.
We carried out this program in \citet{Cons:Laro:2012} for Gaussian DAG models,
using the method of the fractional Bayes factor \citep{Ohag:1995}.
Our findings were consistent with, and extended,
those presented in \citet{Carv:Scot:2009} for Gaussian decomposable graphical models,
which relied on the use of the hyper-inverse Wishart distribution \citep{Leta:Mass:2007}.


More recently, research has shifted towards
\textit{covariate-adjusted} estimation of covariance matrices.
Motivation for this research stems from the analysis of genetical genomics data
(eQTL analysis) where the aim is to study conditional dependence structures
of gene expressions after the confounding genetic effects  
are taken into account.
Indeed, an important finding from many genetical genomics experiments is that
the gene expression level of many genes is inheritable and can be partially explained
by genetic variation; see e.g.~\citet{Brem:Krug:2005}.
Since some genetic variants have effects on the expression of multiple genes,
they act as confounders when trying to learn the association between the genes.
Accordingly, ignoring the effects of genetic variants on the gene expression levels
can lead to both false positive and false negative associations in the gene network graph.
The effect of genetic variants on gene expression therefore needs to be adjusted
in estimating the high-dimensional precision matrix
\citep{Caietal:2013, Chenetal:2013}.

The problem is usually formulated as one of joint estimation
of multiple regression coefficients and a precision matrix,
with the latter assumed to be Markov with respect to some graph.
Since these models are used in high-dimensional settings,
both the regression and the covariance structure are assumed to be sparse.
\citet{Rothetal:2010}, \citet{Yin:Li:2011} and \citet{Chenetal:2013}
assume that the error term is multivariate normal;
this assumption is relaxed in the paper by \citet{Caietal:2013}.
The literature in the area, as exemplified in all the  papers  above,
is  carried out within a constrained minimization approach (under a suitable norm).
Contributions in the Bayesian framework are still very limited.
A notable exception is \citet{Bhad:Mall:2013} who perform variable and covariance
selection jointly, using decomposable graphs and weakly informative hierarchical priors.


In this paper we deal with covariate-adjusted selection of Gaussian DAG models
within an objective Bayes framework.
Specifically, we reconsider the foundations
of the approach by \citet{Geig:Heck:2002},
originally presented for the case of
i.i.d.~sampling,
and show that it can be meaningfully extended to the
multivariate regression setting.
We provide closed-form expressions for the marginal likelihood of any DAG,
then we propose an objective Bayes procedure, based on the fractional Bayes factor,
which works for DAGs with small parent sets.
Our results extend to the regression setup those of \citet{Cons:Laro:2012}
and \citet{Carv:Scot:2009}; they also complement those of \citet{Bhad:Mall:2013},
both because they are derived within an objective framework,
and because they cope with a broader family of graphs.




The paper is organized as follows. Section~\ref{sec:matrixDistrib} reviews
the matrix distributions used in the paper, and section~\ref{sec:multivariateRegression}
presents the Gaussian multivariate regression setup. Section~\ref{sec:ob}
illustrates our objective framework, while section~\ref{sec:covsel} contains
our proposal for covariance selection. Finally, section~\ref{sec:disc}
briefly discusses our work.

\section{Matrix distributions}
\label{sec:matrixDistrib}

Consider $n$ independent observations on $q$ continuous dependent variables,
arranged in an $n \times q$ response matrix:
\begin{equation}
\label{eq:dataMatrix}
\bY=
\left(\begin{array}{c}
  	  \by_{1}^\top\\
      \vdots\\
  	  \by_{n}^\top
  	  \end{array}
\right)=
\left(\begin{array}{ccc}
      \bY_{1} & \ldots & \bY_q\\
      \end{array}
\right),
\end{equation}
where $\by_i=(y_{i1}, \ldots, y_{iq})^\top$ is the $i$-th observation,
and $\bY_j=(y_{1j}, \ldots, y_{nj})^\top$ represents the observations
on the $j$-th variable. Let  $\bX$ be a design matrix  with $n$ rows
and $p+1$ columns ($p$~predictors plus intercept) which we assume known
without error; denote by $\bx_1^\top,\dots,\bx_n^\top$  its rows.
We model the observations as
$\by_i\g\bB,\bSigma\sim\N_q(\bB^\top\bx_i, \bSigma)$,
independently over $i=1,\dots,n$,
where $\bB$ is an unconstrained
$(p+1)\times q$ matrix,
$\bSigma$ is an s.p.d.~$q\times q$ matrix,
and $\N_q(\bmu,\bSigma)$ denotes the $q$-variate normal distribution
with mean vector $\bmu$ and covariance matrix $\bSigma$.
The $j$-th column of $\bB$, namely $\bB_j$, is the vector of regression coefficients
for the $j$-th variable, and $\vat(\bY\g\bB,\bSigma)=\bX\bB$.
The distribution of $\bY$, given $\bB$ and~$\bSigma$,
is a special case of the matrix normal distribution;
the general case, reviewed in section~\ref{subsec:matrixNormal},
will give a conjugate prior for $\bB$ (given $\bSigma)$.
A conjugate prior for $\bSigma^{-1}$ will be given by the Wishart distribution,
which is reviewed in section~\ref{subsec:Wishart}.

\subsection{Matrix normal} 
\label{subsec:matrixNormal}
We say that the random matrix $\bY$ follows
the \emph{matrix normal distribution} with mean matrix $\bM$,
row covariance matrix $\bPhi$, and column covariance matrix $\bSigma$,
when $\mathrm{vec}(\bY)$ follows the multivariate normal distribution
with mean vector $\mathrm{vec}(\bM)$ and covariance matrix
$\bSigma\otimes\bPhi$; recall that $\mbox{vec}(\bY)$ is the vector obtained from~$\bY$
by stacking its columns on top of one another,
while $\otimes$\ denotes the Kronecker product.
If $\bY$ is an $n\times q$ matrix, $\bM$ will be
an $n\times q$ matrix, $\bPhi$ an s.p.d.~$n \times n$ matrix,
$\bSigma$~an s.p.d.~$q \times q$ matrix, and we will write
\begin{equation}
\label{eq:sampling}
\bY\g\bM,\bPhi, \bSigma\sim\N_{n,q}(\bM,\bPhi,\bSigma);
\end{equation}
see \citet[p.~55]{Gupt:Naga:2000}, and \cite{Dawi:1981}, for more information.
We obtain the special case where $\bY$ is a response matrix,
previously described and taken up in section~\ref{sec:multivariateRegression},
by letting $\bM=\bX\bB$, and $\bPhi=\bI_n$,
where $\bI_n$ is the $n\times n$ identity matrix.

Let $\phi_{i i^{\prime}}$ denote the generic element of $\bPhi$,
and $\sigma_{j j^{\prime}}$ the generic element of $\bSigma$.
Clearly, we have $\vat(\bY \g \bM, \bPhi,\bSigma)=\bM$.
Moreover, we have $\cov(y_{ij},y_{i^{\prime}j^{\prime}}\g\bM,\bPhi,\bSigma)=
\phi_{i i^{\prime}}\sigma_{j j^{\prime}}$,
so that $\var(\by_i \g \bM, \bPhi, \bSigma)=\phi_{ii}\bSigma$, $i=1, \ldots, n$,
whereas $\var(\bY_j \g \bM, \bPhi, \bSigma)=\sigma_{jj}\bPhi$, $j=1, \ldots, q$,
with $\var(\bu)$ denoting the covariance matrix of the random vector $\bu$.
More generally, we find
$\cov(\by_i,\by_{i^{\prime}}\g \bM,\bPhi,\bSigma)= \phi_{i i^{\prime}}\bSigma$
and $\cov(\bY_j,\bY_{j^{\prime}} \g \bM,\bPhi,\bSigma)=
\sigma_{j j^{\prime}}\bPhi$,
if we denote by $\cov(\bu,\bv)$ the cross-covariance matrix
of $\bu$ and $\bv$, whose elements are the covariances between all pairs
consisting of one element in $\bu$ and the other in~$\bv$.
Notice that $\cov(\bu,\bu)=\var(\bu)$.

Reparameterizing from $\bSigma$ s.p.d.~to $\bOmega=\bSigma^{-1}$ s.p.d.,
and from $\bPhi$ s.p.d.~to $\bK=\bPhi^{-1}$ s.p.d., which we will find useful
for Bayesian analysis, the density of the matrix normal distribution
$\N_{n,q}(\bM,\bK^{-1},\bOmega^{-1})$ can be written as
\begin{equation}
\label{eq:matrixNormal}
f(\bY \g \bM, \bK,\bOmega)=
\frac{|\bK|^{\frac{q}{2}}| \bOmega|^{\frac{n}{2}}}{(2\pi)^{\frac{nq}{2}}}
\exp\left\{-\frac{1}{2}\tr\left(\bOmega(\bY-\bM)^\top\bK(\bY-\bM)
                         \right)\right \},
\end{equation}
where $|\bPsi|$ denotes the determinant of the matrix $\bPsi$,
and $\tr(\bPsi)$ its trace. Formula~(\ref{eq:matrixNormal}) follows from the density of
$\mathrm{vec}(\bY)\g\mathrm{vec}(\bM),\bOmega^{-1}\otimes\bK^{-1}$,
keeping into account that $\tr(\bOmega\bPsi\bK\bPsi^\top)=\tr(\bPsi^\top\bOmega\bPsi\bK)$
is the value at $(\bPsi,\bPsi)$ of the bilinear form associated to
$\bOmega\otimes\bK=(\bOmega^{-1}\otimes\bK^{-1})^{-1}$,
which is the precision matrix of $\mathrm{vec}(\bY)$,
and that $|\bOmega\otimes\bK|=|\bOmega|^n|\bK|^q$;
see \citet[App.~B]{Laur:1996}.
We call $\bK$ the row precision matrix of~$\bY$,
and $\bOmega$ its column precision matrix.
Clearly, whenever $\bY\g\bM,\bK,\bOmega\sim N_{n,q}(\bM,\bK^{-1},\bOmega^{-1})$,
we have $\bY^\top\g\bM,\bK,\bOmega\sim\N_{q,n}(\bM^\top, \bOmega^{-1}, \bK^{-1})$,
which means $\mathrm{vec}(\bY^\top)\g\bM,\bK,\bOmega\sim
\N_{qn}(\mathrm{vec}(\bM^\top),\bK^{-1}\otimes\bOmega^{-1})$.

Now let $J$ be a proper subset of $\{1,\dots,q\}$,
and denote by $\bY_J$ the submatrix of~$\bY$
consisting of the columns indexed by $J$.
It is immediate to check that $\mathrm{vec}(\bY_J)$ is multivariate normal
with mean vector $\mathrm{vec}(\bM_J)$ and covariance matrix $\bSigma_{JJ}\otimes\bPhi$,
where $\bSigma_{JJ}$ is the submatrix of $\bSigma$
consisting of the rows and columns indexed by $J$;
see \citet[prop.~(C.4)]{Laur:1996}.
Hence, \emph{column marginalization} results in
\ben
\label{eq:marginalization}
\bY_J\g\bM,\bPhi,\bSigma\sim\N_{n,|J|}(\bM_J,\bPhi,\bSigma_{JJ}).
\een
Notice that, if $\bM=\bX\bB$, then $\bM_J=\bX\bB_J$.

Finally, letting $\bar{J}=\{1,\dots,q\}\setminus J$,
it is well known that $\mathrm{vec}(\bY_J)\g\mathrm{vec}(\bY_{\bar{J}})$
is multivariate normal with mean vector
$\mathrm{vec}(\bM_J)-(\bOmega_{JJ}^{-1}\otimes\bK^{-1})(\bOmega_{J\bar{J}}\otimes\bK)
\mathrm{vec}(\bY_{\bar{J}}-\bM_{\bar{J}})$,
and precision matrix $\bOmega_{JJ}\otimes\bK$,
where $\bOmega_{JJ}^{-1}=(\bOmega_{JJ})^{-1}$;
see \citet[prop.~C.5]{Laur:1996}.
Since $(\bOmega_{JJ}^{-1}\otimes\bK^{-1})(\bOmega_{J\bar{J}}\otimes\bK)=
(\bOmega_{JJ}^{-1}\bOmega_{J\bar{J}})\otimes(\bK^{-1}\bK)=
(\bOmega_{JJ}^{-1}\bOmega_{J\bar{J}})\otimes\bI_n$,
we find
\begin{equation}
\label{eq:conditioning}
\bY_J\g\bY_{\bar{J}},\bM,\bK,\bOmega\sim
\N_{n,|J|}(\bM_J-(\bY_{\bar{J}}-\bM_{\bar{J}})\bOmega_{\bar{J}J}\bOmega_{JJ}^{-1},
\bK^{-1},\bOmega_{JJ}^{-1})
\end{equation}
for \emph{column conditioning}.
In the case $\bK=\bI_n$, formula (\ref{eq:conditioning}) returns
$\by_{iJ}\g\bM,\bK,\bOmega\sim
\N_{|J|}(\bem_{iJ}-\bOmega_{JJ}^{-1}\bOmega_{J\bar{J}}(\by_{i\bar{J}}-\bem_{i\bar{J}}),
\bOmega_{JJ}^{-1})$, independently over $i=1,\dots,n$,
where $\by_{iJ}$ and $\bem_{iJ}$ are the subvectors of $\by_i$ and $\bem_i$,
respectively, consisting of the elements indexed by $J$,
while $\bem_i^\top$ is the $i$-th row of $\bM$.

\subsection{Wishart} 
\label{subsec:Wishart}

Let $\bOmega$ be a $q \times q$ \textit{unconstrained} s.p.d.~random matrix.
We will write $\bOmega \sim \W_q(a,\bR)$ to mean that $\bOmega$ follows
a Wishart distribution with density
\begin{equation}
p(\bOmega)=\frac{1}{2^\frac{aq}{2}\Gamma_q(\frac{a}{2})}
|\bR|^{\frac{a}{2}}|\bOmega|^{\frac{a-q-1}{2}}
\exp\left\{-\frac{1}{2}\tr(\bOmega\bR)\right\},
\label{eq:Wishart}
\end{equation}
$\bOmega$ s.p.d., and $p(\bOmega)=0$, otherwise, 
where $\bR$ is a $q \times q$ s.p.d.~matrix,
$a$ is a scalar strictly greater than $q-1$,
and $\Gamma_q(\frac{a}{2})=\pi^{\frac{q(q-1)}{4}}
\prod_{j=1}^q \Gamma\left(\frac{a}{2}+\frac{1-j}{2}\right)$
is the $q$-dimensional gamma function at $a/2$
(generalizing $\Gamma(a/2)=\int_0^\infty z^{\frac{a}{2}-1}e^{-z}dz$).
As for parameter interpretation,
it can be shown that $\vat[\bOmega|\bR,a]=a\bR^{-1}$.
Our notation $\W_q(a,\bR)$ for the density~(\ref{eq:Wishart}) is essentially that of
\citet[p.~59]{DeGr:1970};
other authors \citep{Pres:1982,Laur:1996} would use
$\bR^{-1}$ in place of $\bR$.

We now recall some useful results. 
Let $\bOmega$ be the precision matrix of $\by\g\bmu,\bSigma\sim\N_q(\bmu,\bSigma)$,
that is, $\bOmega=\bSigma^{-1}$. Think of $\by$ as the generic row
of the matrix $\bY$ (dropping subscript $i$).
Partition $\bSigma$ and $\bOmega$ into the blocks
corresponding to the variables indexed by $J$ and its complement $\bar{J}$,
for a given proper subset $J$ of $\{1,\dots,q\}$:
\begin{equation}
    \bSigma=
  \left[
    \begin{array}{cc}
  	\bSigma_{JJ} & \bSigma_{J\bar{J}}\\
  	\bSigma_{\bar{J}J} & \bSigma_{\bar{J}\bar{J}}
  	\end{array}
  \right],\quad
    \bOmega=
  \left[
    \begin{array}{cc}
  \bOmega_{JJ} & \bOmega_{J\bar{J}}\\
  \bOmega_{\bar{J}J} & \bOmega_{\bar{J}\bar{J}}
  \end{array}
  \right].
\label{eq:partitionedOmega}
\end{equation}
The block $\bSigma_{JJ}$ is the marginal covariance matrix of $\by_J$
(obtained from $\by$ by selecting the elements of $\by$ indexed by $J$).
Denote by $\bSigma_{JJ \cdot \bar{J}}$ the conditional covariance matrix
$\var(\by_J \g \by_{\bar{J}})$ of $\by_J$ given $\by_{\bar{J}}$
(obtained from $\by$ by complementary selection).
Then
\begin{equation}
\bSigma_{JJ \cdot \bar{J} } = \bSigma_{JJ}-
\bSigma_{J\bar{J}}\bSigma_{\bar{J}\bar{J}}^{-1}\bSigma_{\bar{J}J}=
\bOmega_{JJ}^{-1},
\label{eq:condVar}
\end{equation}
that is, $\bSigma_{JJ \cdot \bar{J}}$ is the \emph{Schur complement}
of $\bSigma_{\bar{J}\bar{J}}$ in $\bSigma$,
as well as the inverse of $\bOmega_{JJ}$.

Formula (\ref{eq:condVar}) expresses a relationship between four blocks of $\bSigma$
and a corresponding block of $\bSigma^{-1} = \bOmega$. Hence,
by switching the roles of $\bSigma$ and $\bOmega$, we obtain
\begin{equation}
\bSigma_{JJ}= \left(\bOmega_{JJ}-\bOmega_{J\bar{J}}
\bOmega_{\bar{J}\bar{J}}^{-1}\bOmega_{\bar{J}J} \right)^{-1}=
\bOmega_{JJ \cdot \bar{J}}^{-1},
\label{Sigmavv}
\end{equation}
where $\bOmega_{JJ \cdot \bar{J}}^{-1}$ is to be interpreted as
Schur complementation followed by inversion.
Therefore, working with covariance matrices, marginalization corresponds to
submatrix extraction and conditioning to Schur complementation,
whereas, working with precision matrices, marginalization corresponds to
Schur complementation and conditioning to submatrix extraction.

Now let $\bOmega\sim\W_q(a,\bR)$, with $\bR$ an s.p.d.~matrix and $a>q-1$.
If $\bOmega$ is partitioned as in (\ref{eq:partitionedOmega}),
and $\bR$ is partitioned accordingly, then
\ben
\bOmega_{JJ\cdot\bar{J}} &\sim& \W_{|J|}(a-|\bar{J}|,\bR_{JJ}),
\label{eq:distrOmegaMar}
\een
independently of $(\bOmega_{J\bar{J}},\bOmega_{\bar{J}\bar{J}})$,
where of course $|\bar{J}|=q-|J|$; see \citet[prop.~C.15]{Laur:1996}
who also gives the distribution of $(\bOmega_{J\bar{J}},\bOmega_{\bar{J}\bar{J}})$.

\section{Gaussian multivariate regression}
\label{sec:multivariateRegression}

We return to the scenario discussed in the Introduction,
leading to covariate-adjusted graphical model selection,
and to the response matrix $\bY$ introduced at the beginning
of section~\ref{sec:matrixDistrib}.
Denote by $\bZ$ the $n \times p_\star$ matrix of all possible $p_\star$ predictors.
In eQTL analysis $p_\star$ is typically very large,
and often much larger than $n$. However, because of sparsity considerations,
only models of the type $\bY=\bX \bB+\bE$ need be taken into consideration,
where $\bX$ is an $n \times (p+1)$ design matrix having the unit vector $\bone_n$
as first column and $p\ll p_\star$ predictors from $\bZ$ as remaining columns,
while $\bE$ is an $n \times q$ matrix of error terms
with distribution $\N_{n,q}(\bzero, \bI_n,\bOmega^{-1})$,
and $\bB$ is a $(p+1) \times q $ matrix of regression coefficients
($\bzero$ being the $n\times q$ zero matrix).
Hence, in principle, it is not unreasonable to assume $n> p+1$;
in practice $p$  will be much smaller than~$n$,
as we illustrate in the Discussion.
Notice that the $p$ predictors to be used will not be known \emph{a priori},
and therefore it will be necessary to carry out variable selection
together with covariance selection; this will be feasible
using the marginal likelihoods corresponding to different design matrices.
For simplicity, we will use a single $\bX$ in our notation
(without explicitly conditioning on it).

In section~\ref{subsec:conjugateAnalysis} we summarize the main features of
a standard conjugate analysis of the model
\begin{equation}
\label{eq:multivRegrModel}
\bY\g\bB,\bOmega\sim\N_{n,q}(\bX\bB,\bI_n,\bOmega^{-1}),
\end{equation}
with $\bB$ unconstrained.
This is done for completeness and for the benefit of the reader,
so that the subsequent sections can be followed more easily; see also
\citet{Rowe:2003}, whose notation is somewhat different from ours.
Next, in section~\ref{subsec:marginalDataDistrib}, we derive
the marginal data distribution for a subset of variables
(selected columns of~$\bY$) which represents the building block
for computing the marginal likelihood of a general DAG model
(as detailed in section~\ref{subsec:DAG}).
We remark that, because of the theory presented in section \ref{subsec:DAG},
we need only  consider an unconstrained $\bOmega$ even when we deal with
covariance matrices having a graphical structure.
This is indeed a major simplification characterizing the approach
taken in this paper; we will return to this issue later on.

\subsection{Conjugate analysis}
\label{subsec:conjugateAnalysis}

If we denote by $ \hat{\bB}=(\bX^\top\!\bX)^{-1}\bX^\top\bY$
the least squares estimator of $\bB$, the likelihood function can be written as
\begin{equation}
\label{eq:lik}
f(\bY \g \bB,\bOmega)=
\frac{|\bOmega|^{\frac{n}{2}}}{(2\pi)^{\frac{nq}{2}}}
\exp\left\{-\frac{1}{2}\tr\left(\bOmega\{(\bB-\hat{\bB})^\top\bX^\top\!\bX(\bB-\hat{\bB})+
\hat{\bE}^\top\!\hat{\bE}\}
\right) \right \},
\end{equation}
where $\hat{\bE}=(\bY-\bX\hat{\bB})$ is the matrix of residuals.
Hence, a conjugate prior for $(\bB,\bOmega)$
is obtained by letting
\be
\bB \g \bOmega & \sim & \N_{p+1,q}(\underline{\bB}, \bC^{-1}, \bOmega^{-1}),\\
\bOmega & \sim & \W_q(a,\bR),
\ee
which results in the prior density
\begin{equation}
\label{eq:priorDensity}
p(\bB,\bOmega)=
\frac{|\bOmega|^{\frac{(p+1)+(a-q-1)}{2}}}{K(\bC,\bR,a)}
\exp\left\{-\frac{1}{2}\tr\left(
\bOmega\{(\bB-\underline{\bB})^\top\bC(\bB-\underline{\bB})+
\bR\}\right)\right\},
\end{equation}
where
\begin{equation}
\label{eq:priorNormalizConst}
K(\bC,\bR,a)=
\frac{(2\pi)^{\frac{q(p+1)}{2}}2^\frac{aq}{2}\Gamma_q(\frac{a}{2})}
{|\bC|^\frac{q}{2}|\bR|^\frac{a}{2}}
\end{equation}
is the prior normalizing constant.
The prior (\ref{eq:priorDensity}) is a matrix normal Wishart.

Some algebraic manipulations show that the posterior distribution of $(\bB,\bOmega)$ is
\be
\bB \g \bOmega,\bY &\sim&
\N_{p+1,q}(\overline{\bB}, (\bC+\bX^\top\!\bX)^{-1}, \bOmega^{-1} ), \\
\bOmega \g \bY & \sim & \W_q(a+n,\bR+\hat{\bE}^\top\!\hat{\bE}+\bD),
\ee
where $\overline{\bB}=(\bC+\bX^\top\!\bX)^{-1}(\bX^\top \bY +\bC\underline{\bB})$
is the posterior expectation (matrix) of $\bB$, and $\bD=
(\underline{\bB}-\hat{\bB})^\top\{\bC^{-1}+(\bX^\top\!\bX)^{-1}\}^{-1}
(\underline{\bB}-\hat{\bB})$ is a measure of discrepancy
between $\underline{\bB}$ and $\hat{\bB}$ (prior and data).
Prior-to-posterior updating thus takes the form
\begin{equation}
\label{eq:updatingHyperCoeff}
\underline{\bB} \mapsto \overline{\bB},\quad \bC \mapsto \bC+\bX^\top\!\bX,\quad
a \mapsto a+n,\quad \bR \mapsto \bR+\hat{\bE}^\top\!\hat{\bE}+ \bD,
\end{equation}
and the posterior density $p(\bB,\bOmega\g\bY)$ is as in (\ref{eq:priorDensity})
with hyper-parameters updated by~(\ref{eq:updatingHyperCoeff});
the posterior normalizing constant will be given by
\begin{equation}
\label{eq:postNormalizConst}
K(\bC+\bX^\top\!\bX,\bR+\hat{\bE}^\top\!\hat{\bE}+\bD,a+n),
\end{equation}
with the function $K(\cdot,\cdot,\cdot)$ defined in (\ref{eq:priorNormalizConst}).

\subsection{Marginal data distribution}
\label{subsec:marginalDataDistrib}

The marginal distribution of the matrix $\bY$ can be obtained as
$$
m(\bY)=\frac{f(\bY\g\bB,\bOmega)p(\bB,\bOmega)}{p(\bB,\bOmega\g\bY)},
$$
which in light of conjugacy gives
\begin{equation}
\label{marginalDataDist}
m(\bY)=\frac{K(\bC+\bX^\top\!\bX,\bR+\hat{\bE}^\top\!\hat{\bE}+\bD,a+n)}
{(2 \pi)^{\frac{nq}{2}}K(\bC,\bR,a)},
\end{equation}
that is, up to a multiplicative factor, 
the ratio of the posterior and prior normalizing constants,
(\ref{eq:postNormalizConst}) and (\ref{eq:priorNormalizConst}),
respectively.

In the sequel, we will also need the marginal distribution of
selected columns of the data matrix $\bY$,
corresponding to a proper subset $J$ of the full set of $q$ response variables.
Let $\bY_J$ be the $n\times |J|$ selected data submatrix,
and $\bB_J$ be the corresponding $(p+1)\times |J|$ submatrix of $\bB$,
whose columns contain the regression coefficients for the selected responses.
When restricted to the set $J$ of response variables,
by the results presented in section~\ref{sec:matrixDistrib},
the Gaussian multivariate regression model (\ref{eq:multivRegrModel})
can be written as
$$
\bY_J\g\bB_J,\bOmega_{JJ\cdot\bar{J}}\sim
\N_{n,|J|}\left(\bX\bB_J,\bI_n,\bOmega_{JJ\cdot\bar{J}}^{-1}\right),
$$
with induced prior
\be
\bB_J\g\bOmega_{JJ\cdot\bar{J}} &\sim & \N_{p+1,|J|}
\left(\underline{\bB}_J, \bC^{-1}, \bOmega_{JJ \cdot \bar{J}}^{-1}\right),\\
\bOmega_{JJ \cdot \bar{J}} & \sim & \W_{|J|}(a-|\bar{J}|,\bR_{JJ}),
\ee
where $\underline{\bB}_J$ is the appropriate submatrix of $\underline{\bB}$.

One readily sees that the formal structure of model and prior
for a subset $J$ of response variables is the same as for the full data matrix.
As a consequence, the marginal data distribution for the submatrix $\bY_J$
is given by (\ref{marginalDataDist}) with the following substitutions:
$$
q \mapsto |J|,\; \bR \mapsto \bR_{JJ},\; a \mapsto a-|\bar{J}|,\;
\underline{\bB} \mapsto \underline{\bB}_J,\;  \hat{\bB} \mapsto \hat{\bB}_J,\;
\hat{\bE} \mapsto \hat{\bE}_J,\; \bD \mapsto \bD_{JJ},
$$
while $n$, $\bC$ and  $\bX$ remain unchanged.

\section{Objective analysis}
\label{sec:ob}

We assume the reader is familiar with the basic concepts of model selection
from the Bayesian perspective, as described for instance in \citet[ch.~7]{Ohag:Fors:2004}.
Here, in section~\ref{subsec:fractionalbf},
we provide some background on \emph{objective Bayes} model selection,
focusing in particular on a proposal by \citet{Ohag:1995}.
Then, in section \ref{subsec:fractionalMarLik},
we give the expression for the marginal data distribution
of a generic subset of columns of $\bY$ under the prior implied by such proposal;
this will be instrumental in the construction of the marginal likelihood of a DAG model
given in section~\ref{subsec:DAG}.

\subsection{Fractional parameter priors}
\label{subsec:fractionalbf}

Let $\M_1,\dots,\M_K$ be a collection of Bayesian models for the same observable~$\bY$.
Each model $\M_k$, $k=1,\dots,K$, consists of a family of sampling densities
$f_k(\bY\g\btheta_k)$, indexed by a model specific parameter $\btheta_k$,
and of a prior density $p_k(\btheta_k)$ on $\btheta_k$,
which we assume to be \textit{proper}.
We focus on the comparison of $\M_k$ with $\M_{k^\prime}$
through the Bayes factor.
The Bayes factor for $\M_k$ against $\M_{k^\prime}$ is defined as
$BF_{kk^\prime}(\bY)=m_k(\bY)/m_{k^\prime}(\bY)$,
where $m_k(\bY)=\int f_k(\bY\g\btheta_k)p_k(\btheta_k) d\btheta_k$
is the marginal density of $\bY$ under $\M_k$,
also known as the marginal likelihood of $\M_k$.

In lack of substantive prior information,
we would like to take $p_k(\btheta_k)=p_k^D(\btheta_k)$
for some objective default (non-informative) parameter prior $p_k^D(\btheta_k)$.
However, objective priors are often improper
and they cannot be naively used to compute Bayes factors,
even when the marginal likelihoods $m_k(\bY)$ are finite and non-zero,
because of the presence of arbitrary constants
which do not cancel out in their ratios.
\cite{peri:2005} reviews several proposals put forward to address this issue.
In this paper, we take advantage of the fractional Bayes factor
originally introduced by \cite{Ohag:1995}; see also \citet{Ohag:Fors:2004}.

Let $b=b(n)$, $0<b<1$, be a fraction of the number of observations $n$. Define
\ben
\label{fractionalMarginalLik}
m_k(\bY;b)=\frac{\int f_k(\bY\g\btheta_k)p_k^D(\btheta_k)d\btheta_k}
{\int f_k^b(\bY\g\btheta_k)p_k^D(\btheta_k)d\btheta_k},
\een
where $f_k^b(\bY\g\btheta_k)$ is the sampling density under model $\M_k$
raised to the $b$-th power, and the two integrals are assumed to be finite and non-zero.
The \emph{fractional marginal likelihood}~(\ref{fractionalMarginalLik}) of model $\M_k$,
can be rewritten as
$$
m_k(\bY;b)=\int f_k^{1-b}(\bY\g\btheta_k)p_k^F(\btheta_k\g b,\bY)d\btheta_k,
$$
where $p_k^F(\btheta_k\g b,\bY) \propto f_k^b(\bY\g\btheta_k)p_k^D(\btheta_k)$
is the implied \emph{fractional prior} (actually a ``posterior'' based on
the fractional likelihood and the default prior).
The fractional Bayes factor for $\M_k$ against $\M_{k^\prime}$ is then defined as
the ratio of $m_k(\bY;b)$ to $m_{k^\prime}(\bY;b)$.
In essence, a fraction of the data is used to obtain a proper prior,
which is then applied to the remaining fraction.

Clearly, the fractional prior depends on the choice of $b$.
Usually $b$ will be small, so that dependence of the prior on the data will be weak.
Consistency is achieved as long as $b \rightarrow 0$ for $n \rightarrow \infty$.
\citet[sect.~4]{Ohag:1995} suggests $b=n_0/n$ as a default choice,
where $n_0$ is the minimal (integer) training sample size
for which the fractional marginal likelihood is well defined,
together with a couple of alternative choices,
to be used when robustness is an issue. \cite{More:1997} has an argument
according to which the default choice is the only valid one,
and we stick to this choice in this paper.


\subsection{Fractional marginal likelihoods}
\label{subsec:fractionalMarLik}

Consider the Gaussian multivariate regression model (\ref{eq:multivRegrModel}).
We start from the prior
\ben
\label{pDOmega}
p^D(\bB,\bOmega) \propto |\bOmega|^{\frac{a_D-q-1}{2}},
\een
$\bOmega$ s.p.d., which is flexible enough to accommodate different choices
of default priors. In particular, $a_D=0$ gives
$p^D(\bB,\bOmega) \propto |\bOmega|^{\frac{-(q+1)}{2}}$,
equivalently $p^D(\bB,\bSigma) \propto |\bSigma|^{\frac{-(q+1)}{2}}$
for $\bSigma=\bOmega^{-1}$, because the Jacobian of $(\bB,\bOmega) \mapsto (\bB,\bSigma)$
is proportional to $|\bSigma|^{−(q+1)}$.
This is the ``independence'' Jeffreys prior, that is,
the prior obtained by multiplying the Jeffreys priors for the two parameters
assuming the other one is known;
see \citet[sect.~3.6.2 and (14.2.7)]{Pres:1982}.
Alternatively, $a_D=q-1$ gives
$p^D(\bB,\bOmega)\propto|\bOmega|^{-1}$,
or $p^D(\bB,\bSigma) \propto |\bSigma|^{-q}$.
Both these priors are discussed in \citet{Geis:Corn:1963},
whereas \citet{Geis:1965} focusses more deeply on the independence Jeffreys prior.
\citet{Sun:Berg:2007} present further objective priors for the multivariate normal model.

%
%
%

Using the default prior (\ref{pDOmega}),
and setting the fraction $b$ equal to $n_0/n$,
the fractional prior for the multivariate regression model (\ref{eq:multivRegrModel})
is given by
\begin{equation}
\label{FBFPrior}
p(\bB,\bOmega)\propto|\bOmega|^{\frac{a_D+n_0-p-q-2}{2}}\exp\left\{-\frac{n_0}{2}
\tr\left(\bOmega\left\{(\bB-\hat{\bB})^\top\tilde{\bC}\,(\bB-\hat{\bB})+\tilde{\bR}\right\}
\right)\right\},
\end{equation}
where $\tilde{\bC}=n^{-1}\bX^\top\!\bX$ and $\tilde{\bR}=n^{-1}\hat{\bE}^\top\!\hat{\bE}$;
this is clearly a matrix normal Wishart, having the form (\ref{eq:priorDensity}) with
$$
\underline{\bB}=\hat{\bB},\quad \bC=n_0\tilde{\bC},\quad
a=a_D+n_0-p-1,\quad\bR=n_0\tilde{\bR}.
$$
The prior (\ref{FBFPrior}) is proper under two conditions:
i) $a_D+n_0-p>q$, so that $a>q-1$;
ii) $n-p-1>q-1$, so that $\hat{\bE}^\top\!\hat{\bE}$ is (almost surely) positive definite.

Condition~ii), which simplifies to $n>p+q$,
will not be met in our intended application setting,
but we will be able to relax it in the context of sparse DAG models;
see section~\ref{subsec:DAG}. Condition~i) becomes $n_0>p+q$, if~$a_D=0$,
or $n_0>p+1$, if~$a_D=q-1$.
Clearly, the fraction $b=n_0/n$ must be larger when using the independence Jeffreys prior,
rather than the prior presented in \citet{Geis:Corn:1963},
especially if $q$ is much larger than~$1$.
Since the fraction of the data to be used should be as small
as possible, we recommend setting  $a_D=q-1$ (and $n_0=p+2$, so that $a=q$).
Notice that, for $b=n_0/n$ to be small, with $n_0>p+1$, we need $p<<n$,
which is a stronger requirement than assuming $n>p+1$
as in section~\ref{sec:multivariateRegression}.
However, as anticipated in section~3, and illustrated in the Discussion,
this requirement will be typically satisfied in our intended application setting.

Posterior updating of the hyper-parameters leads to
$$
\overline{\bB}=\hat{\bB},\quad \bC \mapsto n\tilde{\bC},
\quad a \mapsto a_D+n-p-1,\quad \bR \mapsto n\tilde{\bR},
$$
keeping into account that the fractional prior is to be used on
the likelihood raised to the $(1-b)$-th power, that is, on data with the same $\hat{\bB}$,
$\tilde{\bC}$ and $\tilde{\bR}$, but with $n-n_0$ in place of $n$.
Consequently, using (\ref{marginalDataDist}), one gets
$$
m(\bY)=
\frac{
K(\bX^\top\!\bX,\hat{\bE}^\top\!\hat{\bE},a_D+n-p-1)}
{(2 \pi)^\frac{nq}{2}
K(n_0n^{-1}\hat{\bX}^\top\!\hat{\bX},n_0n^{-1}\hat{\bE}^\top\!\hat{\bE},a_D+n_0-p-1)
},
$$
which after some simplifications leads to
\ben
\label{marginalDataDistFBF}
m(\bY)=
\pi^{-\frac{(n-n_0)q}{2}}
\frac{\Gamma_q(\frac{a_D+n-p-1}{2})}{\Gamma_q(\frac{a_D+n_0-p-1}{2})}
\left(
\frac{n_0}{n}
\right)^{\frac{q(a_D+n_0)}{2}}
|\hat{\bE}^\top\!\hat{\bE}|^{-\frac{n-n_0}{2}}.
\een

In order to apply the method presented in section \ref{sec:covsel}
one also needs the fractional marginal likelihood based on the submatrix $\bY_J$
which only contains the columns of~$\bY$ belonging to the subset $J$,
which we write as $m(\bY_J)$.
This marginal likelihood is germane to our approach,
and represents a half-way house towards computing the entire fractional marginal likelihood
for a DAG model; see section \ref{subsec:DAG}.
Based on the results presented in section \ref{subsec:marginalDataDistrib},
it is immediate to conclude that $m(\bY_J)$ can be obtained
from equation~(\ref{marginalDataDistFBF}) upon making the substitutions
$$
q \mapsto |J|,\quad a_D \mapsto a_D-|\bar{J}|,
\quad\hat{\bE} \mapsto \hat{\bE_J}=(\bY_J-\bX\hat{\bB}_J).
$$
These substitutions lead to
\ben
\label{marginalDataDistFBFSubsetMatrix}
m(\bY_J)=
\pi^{-\frac{(n-n_0)|J|}{2}}
\frac{\Gamma_{|J|}\left(\frac{a_D+n-p-1-|\bar{J}|}{2}\right)}
{\Gamma_{|J|}\left(\frac{a_D+n_0-p-1-|\bar{J}|}{2}\right)}
\left(
\frac{n_0}{n}
\right)^{\frac{|J|(a_D+n_0-|\bar{J}|)}{2}}
|\hat{\bE}_J^\top\!\hat{\bE_J}|^{-\frac{n-n_0}{2}},
\een
which returns (\ref{marginalDataDistFBF}) upon setting $J=\{1,\dots,q\}$.

Formula~(\ref{marginalDataDistFBFSubsetMatrix}) derives from
$\bOmega_{JJ\cdot\bar{J}}\sim\W_{|J|}(a_J,\bR_{JJ})$ with $a_J=a_D+n_0-p-1-|\bar{J}|$,
which is (almost surely) proper if $n>p+|J|$.
The latter condition guarantees positive definiteness of $\bR_{JJ}$,
while $a_J=q-|\bar{J}|=|J|$ using our recommended choices for $a_D$ and $n_0$.
Therefore, formula~(\ref{marginalDataDistFBFSubsetMatrix}) provides us with
a valid value for $m(\bY_J)$, whenever $|J|<n-p$, even if $n\le p+q$.
We will exploit this fact in section~\ref{subsec:DAG}
to derive the marginal likelihood of a sparse DAG.
In the next paragraph we specialize~(\ref{marginalDataDistFBFSubsetMatrix})
to the simplest regression setup, which is of some interest in its own right.

If the sampling distribution corresponds to i.i.d.~observations
from a $q$-dimensional Gaussian density with expectation $\bmu $ and precision $\bOmega$,
conditionally on $\bmu$ and $\bOmega$, the corresponding marginal data distribution 
$m(\bY_J)$ can be derived from (\ref{marginalDataDistFBFSubsetMatrix})
upon setting $p=0$ (no predictors) and $\hat{\bE}=\bY-\bone_n \bar{\by}^\top$,
where $\bar{\by}$ is the $q$-dimensional vector of sample means.
In this way we obtain
\ben
\label{marginalDataDistFBFsubsetIID}
m(\bY_J)=
\pi^{-\frac{(n-n_0)|J|}{2}}
\frac{\Gamma_{|J|}\left(\frac{a_D+n-1-|\bar{J}|}{2}\right)}
{\Gamma_{|J|}\left(\frac{a_D+n_0-1-|\bar{J}|}{2}\right)}
\left(
\frac{n_0}{n}
\right)^{\frac{|J|(a_D+n_0-|\bar{J}|)}{2}}
|\hat{\bE}_J^\top\!\hat{\bE_J}|^{-\frac{n-n_0}{2}},
\een
with $(\hat{\bE}^\top\!\hat{\bE})_{j j^{\prime}}=
\sum_i (y_{ij}-\bar{y}_j)(y_{ij^{\prime}}-\bar{y}_{j^{\prime}})$.
Expression~(\ref{marginalDataDistFBFsubsetIID})
complements formula~(22) in \citet{Cons:Laro:2012},
which holds for i.i.d.~$q$-dimensional Gaussian observations
with zero expectation.


%
%

\section{Covariance selection}
\label{sec:covsel}

So far we have analyzed the Gaussian multivariate regression
model~(\ref{eq:multivRegrModel}) under the condition that $\bOmega$ is unconstrained.
We now assume instead that $\bOmega$ is constrained by a DAG,
aiming at graphical model (or covariance) selection
after having adjusted for the presence of covariates.
In section~\ref{subsec:DAG},
we develop an extension of the approach by \citet{Geig:Heck:2002}
explicitly for the regression setup.
An advantage of the method we present is that
the computation of the marginal likelihood for each DAG
only requires the results established, for an unconstrained $\bOmega$,
in section~\ref{subsec:fractionalMarLik}.
In section~\ref{subsec:UG decomposable},
taking advantage of the fact that any two Markov equivalent DAGs
obtain the same marginal likelihood,
we specify our results to the case of Gaussian decomposable graphical models,
and relate them to those obtained by \citet{Carv:Scot:2009}
in the i.i.d.~case.


\subsection{Acyclic directed error structure}
\label{subsec:DAG}

Let $\D$ be a DAG with vertex set $\{1,\dots,q\}$.
Denote by $\pa_\D(j)$ the \textit{parents} of~$j$ in~$\D$,
that is, the set of all vertices in $\D$ from which an edge points to vertex~$j$,
and by $\by_{i\pa_\D(j)}$ the subvector of $\by_i$ indexed by ${\pa_\D(j)}$.
The multivariate normal sampling density of $\by_i \g \bB,\bOmega$,
assumed to be Markov with respect to $\D$, can be written as
\ben
\label{jointDensity}
f_\D(\by_i \g \btheta_\D)=\prod_{j=1}^q f_\D(y_{ij} \g \by_{i\pa_\D(j)}; \btheta_j),
\een
where $\btheta_j=(\balpha_j,\bgamma_j,\lambda_j)$ is defined by
\ben
\vat(y_{ij}\g\by_{i\pa_\D(j)};\bB,\bOmega)&=&
\bx_i^\top\balpha_j+\by_{i\pa_\D(j)}^\top\bgamma_j,\label{cond:mean}\\
\var(y_{ij}\g\by_{i\pa_\D(j)};\bB,\bOmega)&=&\lambda_j^{-1},\label{cond:var}
\een
and $\btheta_\D=(\btheta_1,\dots,\btheta_q)$ is the collection of all $\btheta_j$s;
recall that $\bx_i^\top$ is the $i$-th row of the design matrix $\bX$,
and notice that we drop dependence on $\D$ when we move from $\btheta_\D$
to its components (to lighten notation).
We illustrate below the reparameterization from $(\bB,\bOmega)$, with $\bOmega$ s.p.d.,
to $\btheta_\D$, with $\lambda_j>0$, $j=1,\dots,q$,
after a remark on~($\ref{jointDensity}$).

The conditional vertex density $f_\D(y_{ij}\g\by_{i\pa_\D(j)}; \btheta_j)$
is a univariate normal density with expectation and variance
given by~(\ref{cond:mean}) and~(\ref{cond:var}), respectively.
It is important to remark that such density depends on $\D$
only through $\pa_\D(j)$. In other words, if two DAGs $\D_1$ and $\D_2$
are such that $\pa_{\D_1}(j)=\pa_{\D_2}(j)$,
then the vertex-specific parameter $\btheta_j$ varies in the same space
under $\D_1$ and $\D_2$, because $\bgamma_j$ has the same dimension under the two DAGs,
and $f_{\D_1}(y_{ij} \g \by_{i\pa_{\D_1}(j)}; \btheta_j)=
f_{\D_2}(y_{ij} \g \by_{i\pa_{\D_2}(j)}; \btheta_j)$.
This property, called \emph{likelihood modularity} by \cite{Geig:Heck:2002},
represents a condition to be satisfied for the subsequent theory to apply.

Assume (without loss of generality) that the vertices of $\D$ are well-numbered;
this means that, if $j^\prime$ is a parent of $j$, then $j^\prime<j$.
If $\D$ is \textit{complete}, that is, it has all pairs of vertices joined by an edge,
then the parameters indexing the last ($j=q$) conditional vertex density
in~(\ref{jointDensity}) are:
$\balpha_q=\bB_q+\bB_{\bar{q}}\bOmega_{\bar{q}q}\Omega_{qq}^{-1}$,
$\bgamma_q=-\bOmega_{\bar{q}q}\Omega_{qq}^{-1}$, and $\lambda_q=\Omega_{qq}$,
where $\bar{q}=\{1,\dots,q-1\}=\pa_\D(q)$;
see the end of section~\ref{subsec:matrixNormal}.
Then, since $\by_{i\bar{q}}\g \bB,\bOmega\sim
\N_{q-1}(\bB_{\bar{q}}^\top\bx_i,\bOmega_{\bar{q}\bar{q}.q}^{-1})$,
one can repeat the previous argument and recursively find $\btheta_{q-1},\dots,\btheta_1$.
If $\D$ is \emph{incomplete}, its missing edges will impose
on~$\btheta_1,\dots,\btheta_q$ the constraints
$\gamma_{jj^\prime}=0$, $j^\prime\notin\pa_\D(j)$, $j=1,\dots,q$,
so that a corresponding set of constraints will be imposed on $\bOmega$.

We now show that, for complete DAGs, the transformation $(\bB,\bOmega)\mapsto\btheta_\D$
is a smooth bijection. This fact, which is arguably not new,
is reported here because it will be used below for constructing priors under general DAGs.
Given the recursive definition of~$(\bB,\bOmega)\mapsto\btheta_\D$,
it is enough to show that the transformation from
$(\bB,\bOmega)$, with $\bOmega$ s.p.d.,
to $(\bB_{\bar{q}},\bOmega_{\bar{q}\bar{q}\cdot q};\balpha_q,\bgamma_q,\lambda_q)$,
with $\bOmega_{\bar{q}\bar{q}\cdot q}$ s.p.d.~and $\lambda_q>0$, is a smooth bijection.
We do this by composing a few simpler reparameterizations.
First, we go from $(\bB,\bOmega)$, with $\bOmega$~s.p.d.,
to $(\bB,\bOmega_{\bar{q}\bar{q}\cdot q},\bOmega_{\bar{q}q},\Omega_{qq})$,
with $\bOmega_{\bar{q}\bar{q}\cdot q}$ s.p.d.~and $\Omega_{qq}>0$,
where the smooth inverse map is provided by
$\bOmega_{\bar{q}\bar{q}}=\bOmega_{\bar{q}\bar{q}\cdot q}+
\bOmega_{\bar{q}q}\Omega_{qq}^{-1}\bOmega_{\bar{q}q}^\top$,
recalling that $\bOmega_{q\bar{q}}=\bOmega_{\bar{q}q}^\top$
(unconstrained); see for instance \citet[Lemma~B.1]{Laur:1996}.
Then, we~trivially split $\bB$ as $(\bB_q,\bB_{\bar{q}})$,
and replace $\bB_q$ with $\balpha_q$,
where the smooth inverse map is given by
$\bB_q=\balpha_q-\bB_{\bar{q}}\bOmega_{\bar{q}q}\Omega_{qq}^{-1}$.
Finally, we reparameterize from $\bOmega_{\bar{q}q}$ to~$\bgamma_q$,
with smooth inverse map given by $\bOmega_{\bar{q}q}=-\Omega_{qq}\bgamma_q$,
and we rename $\Omega_{qq}$ as $\lambda_q$ (constrained to be positive).

In light of the above discussion, all complete DAGs define the same statistical model,
in which $\bOmega$ is unconstrained, and there is a smooth bijection
between their collections of parameters; in the terminology of \citet{Geig:Heck:2002}
we have \emph{complete model equivalence}, and \emph{regularity}.
It follows that any prior on $(\bB,\bOmega)$ will induce
a prior on $\btheta_\D$, if $\D$ is complete. We now show that,
if we let $(\bB,\bOmega)$ follow the conjugate prior (\ref{eq:priorDensity}),
then $p_\D(\btheta_\D)=\prod_{j=1}^qp_\D(\btheta_j)$,
so that $\btheta_1,\dots,\btheta_q$ will be \textit{a priori} independent.
This property is called  \emph{global parameter independence},
and represents a crucial ingredient in the approach of \citet{Geig:Heck:2002};
it can be obtained by recursive application of the following result.

\begin{prop}
If $\bB\g\bOmega\sim\N_{(p+1)\times q}(\underline{\bB}, \bC^{-1}, \bOmega^{-1})$
and $\bOmega\sim\W_q(a,\bR)$, then the pair
$(\bB_{\bar{q}},\bOmega_{\bar{q}\bar{q}\cdot q})$
is independent of the triple
$(\bB_q+\bB_{\bar{q}}\bOmega_{\bar{q}q}\Omega_{qq}^{-1},\bOmega_{\bar{q}q},\Omega_{qq})$.
\end{prop}
\begin{proof}
Consider the reparameterization in terms of $\bOmega_{\bar{q}\bar{q}\cdot q}$ s.p.d.,
$\bOmega_{\bar{q}q}$, $\Omega_{qq}>0$, $\bB_{\bar{q}}$,
$\balpha_q=\bB_q+\bB_{\bar{q}}\bOmega_{\bar{q}q}\Omega_{qq}^{-1}$,
and factorize the corresponding joint parameter density as
$$
p(\balpha_q\g\bB_{\bar{q}},\bOmega_{\bar{q}\bar{q}\cdot q},\bOmega_{\bar{q}q},\Omega_{qq})
\times p(\bB_{\bar{q}}\g\bOmega_{\bar{q}\bar{q}\cdot q},\bOmega_{\bar{q}q},\Omega_{qq})
\times p(\bOmega_{\bar{q}\bar{q}\cdot q},\bOmega_{\bar{q}q},\Omega_{qq}).
$$
We know, from our statement following (\ref{eq:distrOmegaMar}),
that $\bOmega_{\bar{q}\bar{q}\cdot q}$
is independent of $(\bOmega_{\bar{q}q},\Omega_{qq})$
under the assumed distribution for $\bOmega$.
Moreover, from the law of~$\bB\g\bOmega$,
we obtain
\be
\bB_{\bar{q}}\g
\bOmega_{\bar{q}\bar{q}\cdot q},\bOmega_{\bar{q}q},\Omega_{qq}&\sim&
\N_{(p+1),(q-1)}(\underline{\bB}_{\bar{q}},
\bC^{-1},\bOmega_{\bar{q}\bar{q}\cdot q}^{-1}),\\
\balpha_q\g
\bB_{\bar{q}},\bOmega_{\bar{q}\bar{q}\cdot q},\bOmega_{\bar{q}q},\Omega_{qq}&\sim&
\N_{p+1}(\underline{\bB}_q-\underline{\bB}_{\bar{q}}\bOmega_{\bar{q}q}\Omega_{qq}^{-1},
\Omega_{qq}^{-1}\bC^{-1}),
\ee
first using column marginalization (\ref{eq:marginalization}), and (\ref{Sigmavv}),
then using column conditioning (\ref{eq:conditioning}).
Therefore, the joint density of $\bOmega_{\bar{q}\bar{q}\cdot q}$,
$\bOmega_{\bar{q}q}$, $\Omega_{qq}$, $\bB_{\bar{q}}$, and $\balpha_q$,
factorizes as
$$
p(\balpha_q\g\bOmega_{\bar{q}q},\Omega_{qq})
\times p(\bB_{\bar{q}}\g\bOmega_{\bar{q}\bar{q}\cdot q})
\times p(\bOmega_{\bar{q}\bar{q}\cdot q})
\times p(\bOmega_{\bar{q}q},\Omega_{qq}),
$$
which implies the desired result.
\end{proof}

If $\D$ is incomplete, global parameter independence
can be guaranteed by letting $p_\D(\btheta_\D)=\prod_{j=1}^q p_{\C_j}(\btheta_j)$,
where $\C_j$ is any complete DAG such that $\pa_{\C_j}(j)=\pa_\D(j)$.
The actual choice of each $\C_j$ is immaterial,
because all $j^\prime\notin\pa_\D(j)$, $j^\prime\neq j$,
must follow $j$ in $\C_j$, and thus $p_{\C_j}(\btheta_j)$ is induced by
the law of $(\bB_F,\bOmega_{FF\cdot\bar{F}})$,
where $F=\fa_\D(j)=\pa_\D(j) \cup \{j \}$ is the \textit{family} of $j$ in $\D$.
Notice that $j$ goes necessarily last in $\fa_\D(j)$,
and recall that $\bB_F\g\bOmega_{FF\cdot\bar{F}}\sim
\N_{(p+1)\times|F|}(\underline{\bB}_F, \bC^{-1}, \bOmega_{FF\cdot\bar{F}}^{-1})$,
by column marginalization, while $\bOmega_{FF\cdot\bar{F}}\sim\W_{|F|}(a-|F|,\bR_{FF})$,
as per~(\ref{eq:distrOmegaMar}).
%
%
Assigning parameter priors in this way,
we also guarantee \emph{prior modularity}:
$p_{\D_1}(\btheta_j)=p_{\D_2}(\btheta_j)$, if $\pa_{\D_1}(j)=\pa_{\D_2}(j)$.
This is the last ingredient required by the method of \citet{Geig:Heck:2002}
to compute the marginal likelihood of \textit{any} DAG model,
based on the assignment of the \textit{single} prior (\ref{eq:priorDensity}).
We now detail the computations for our regression setting.

The marginal density of the  matrix $\bY$ under the DAG $\D$,
equivalently the marginal likelihood of $\D$ observing $\bY$,
can be found as $m_\D(\bY)=\int f_\D(\bY\g\btheta_\D)p_\D(\btheta_\D)d\btheta_D$,
where $f_\D(\bY\g\btheta_\D)=\prod_{i=1}^nf_\D(\by_i \g \btheta_\D)$
with $f_\D(\by_i\g\btheta_\D)$ given by (\ref{jointDensity}),
and furthermore $p_\D(\btheta_\D)=\prod_{j=1}^qp_\D(\btheta_j)$
by global parameter independence. We can thus write
\be
m_\D(\bY)&=&
\prod_{j=1}^q\int p_\D(\btheta_j)\prod_{i=1}^n
f_\D(y_{ij}\g\by_{i\pa_\D(j)};\btheta_j)d\btheta_j\\
&=&
\prod_{j=1}^q\int p_{\C_j}(\btheta_j)\prod_{i=1}^n
f_{\C_j}(y_{ij}\g\by_{i\pa_{\C_j}(j)};\btheta_j)d\btheta_j\\
&=&
\prod_{j=1}^q\int p_{\C_j}(\btheta_j)
f_{\C_j}(\bY_{j}\g \bY_{\pa_{\C_j}(j)};\btheta_j)d\btheta_j,
\ee
where the second equality is  based on prior and likelihood modularity.
It follows that
\ben
\label{formula18OfGHnew}
m_\D(\bY)=\prod_{j=1}^q m_{\C_j}(\bY_j\g\bY_{\pa_{\C_j}(j)})=
\prod_{j=1}^q\frac{m_{\C_j}(\bY_{\fa_{\C_j}(j)})}{m_{\C_j}(\bY_{\pa_{\C_j}(j)})}=
\prod_{j=1}^q\frac{m(\bY_{\fa_{\D}(j)})}{m(\bY_{\pa_{\D}(j)})},
\een
recalling that $\pa_{\C_j}(j)\equiv\pa_\D(j)$, by construction,
and $m_{\C_j}(\cdot)$ is nothing else but $m(\cdot)$
under our prior (\ref{eq:priorDensity}), by complete model equivalence and regularity.

The great advantage of (\ref{formula18OfGHnew}) is that
the computations of the required terms in the rightmost product
can be done under the assumption that the precision matrix $\bOmega$ is unconstrained,
and thus one can use the standard matrix normal Wishart prior~(\ref{eq:priorDensity}).
Notice that the DAG~$\D$ enters (\ref{formula18OfGHnew}) only through the specification
of the set of parents, $\pa_{\D}(j)$, for each vertex~$j$.
The expressions for $m(\bY_{\fa_\D(j)} )$ and $m(\bY_{\pa_\D(j)})$
are available in section~\ref{subsec:marginalDataDistrib},
upon replacing $J$ with ${\fa_\D(j)}$ and ${\pa_\D(j)}$, respectively.

Prior~(\ref{eq:priorDensity}) requires to specify the
hyper-parameters $\underline{\bB}$, $\bC$, $a$, and~$\bR$.
This can be problematic, especially when the dimension of the problem is large,
and we know that marginal likelihoods are quite sensitive
to changes in the hyper-parameters; see \citet[Ch.~7]{Ohag:Fors:2004}.
We therefore suggest an objective choice,
based on the fractional matrix normal Wishart prior~(\ref{FBFPrior})
applied to the Gaussian likelihood~(\ref{eq:lik}) with $(n-n_0)$ observations
and the same $\hat{\bB}$, $\tilde{\bC}$ and $\tilde{\bR}$ as the data.
With this choice, the terms $m(\bY_{\fa_\D(j)})$ and $m(\bY_{\pa_\D(j)})$
in formula~(\ref{formula18OfGHnew}) can be computed
from~(\ref{marginalDataDistFBFSubsetMatrix})
provided that the condition $|\fa_\D(j)|=|\pa_\D(j)|+1<n-p$ is satisfied.
This condition guarantees a valid value for
$m(\bY_j|\bY_{\pa_\D(j)})=m(\bY_{\fa_\D(j)})/m(\bY_{\pa_\D(j)})$
by granting a proper distribution to the marginal precision matrix
$\bOmega_{\fa_\D(j)\fa_\D(j)\cdot\overline{\fa_\D(j)}}$;
see section~\ref{subsec:fractionalMarLik}.
In this way, formula (\ref{formula18OfGHnew}) provides us with
a valid marginal likelihood (product of $q$ valid conditional marginal likelihoods
given parent observations) for every DAG $\D$ whose parent sets have size smaller than
the number of observations minus the number of columns in the design matrix $\bX$
(number of predictors in the model plus one).
The latter is a sparsity condition on the structure of the DAG,
involving the maximal number of parents across vertices,
which is quite reasonable in our intended application setting
(eQTL analysis) as discussed in the Introduction.

%
%
%

\subsection{Decomposable error structure}
\label{subsec:UG decomposable}

It is often appropriate to model the conditional independence structure
of a set of variables in terms of an undirected graph;
see \citet{Laur:1996} for an authoritative exposition.
This is for instance the approach followed in  \citet{Caietal:2013}
and \citet{Chenetal:2013} for the analysis of genetical genomics data.
With reference to the Gaussian multivariate regression model (\ref{eq:multivRegrModel}),
this means that the precision matrix $\bOmega$ of the response vector $\by_i$
is constrained by an undirected graph $\G$:
if an edge is missing between $j$ and $j^{\prime}$ in $\G$,
then $\bOmega_{j j^{\prime}}=0$. Equivalently, $\by_i$ is Markov with respect to~$\G$,
that is, if $j$ and $j^\prime$ are not joined by an edge in~$\G$,
the responses $y_{ij}$ and~$y_{ij^{\prime}}$ are conditionally independent,
under the sampling distribution, given all remaining responses;
in symbols $y_{ij} \ind y_{ij^{\prime}}\g
\by_{i(\{1,\dots,q\}\setminus\{j,j^{\prime}\})}, \bB, \bOmega$
\citep{Drto:Perl:2004}.

To enhance tractability, the undirected graph $\G$ is often assumed to satisfy
some conditions, such as \textit{decomposability}; see for instance \citet{Bhad:Mall:2013}.
It is well known that a decomposable $\G$ is Markov equivalent
to some DAG \citep{Ande:Madi:Perl:1997}.
Specifically, one can always well-number the vertices of $\G$
and construct a directed version $\G^<$,
which is a DAG  Markov equivalent to $\G$; see \citet[p.~18]{Laur:1996}.
It follows that the methodology developed in section~\ref{subsec:DAG}
can also be applied to decomposable graphs,
because the marginal likelihoods given by such methodology
are invariant with respect to Markov equivalence.
Indeed, the proof of Theorem~4 in~\cite{Geig:Heck:2002}
directly carries over into our regression setting.


In practice, the marginal likelihood of the model defined by the decomposable graph $\G$,
$m_\G(\bY)=m_{\G^<}(\bY)$, will be given by (\ref{formula18OfGHnew}) with $\D=\G^<$.
Since the parameter prior used to compute (\ref{formula18OfGHnew})
satisfies global parameter independence,
$m_{\G^<}(\bY)$ is readily seen to be $\G^<$-Markov;
see for instance \citet[sect.~9.4]{Cowe:Dawi:Laur:Spie:1999}.
Then $m_{\G}(\bY)$ is also $\G$-Markov,
and thus it admits the representation
   \ben
   \label{factorizationUG}
   m_\G(\bY)=\frac{\prod_{C \in {\cal C}}
   m(\bY_C)}{\prod_{S \in {\cal S}} m(\bY_S)},
   \een
where ${\cal C}$ is the set of cliques, and  ${\cal S}$ the set of separators,
of the decomposable graph~$\G$; see \citet{Laur:1996}.
The explicit expression of each factor appearing in (\ref{factorizationUG})
can be deduced from (\ref{marginalDataDist})
as explained in section~\ref{subsec:marginalDataDistrib}.

In particular, when using the fractional matrix normal Wishart prior~(\ref{FBFPrior}),
one computes $m(\bY_C)$ and $m(\bY_S)$ in~(\ref{factorizationUG})
by means of~(\ref{marginalDataDistFBFSubsetMatrix}), with $J=C$ and~$J=S$, respectively,
assuming $|C|<n-p$ (hence $|S|<n-p$) whenever $C$ is a clique
($S\subseteq C$ a separator) of~$\G$.
In this way, we cope with decomposable graphs
whose clique sizes are smaller than the number of observations
minus the number of predictors in the model.
This is again a sparsity assumption on the graph,
well-suited to our intended application setting,
which grants a proper distribution to
$\bOmega_{CC\cdot\bar{C}}$ (hence to $\bOmega_{SS\cdot\bar{S}}$);
see section~\ref{subsec:fractionalMarLik}.
%
%
We remark that formulae~(\ref{factorizationUG}) and~(\ref{marginalDataDistFBFSubsetMatrix})
generalize to the multivariate regression setup
the results established by \citet{Carv:Scot:2009} for i.i.d.~Gaussian observations
with zero expectation. As a special case,
formulae~(\ref{factorizationUG}) and~(\ref{marginalDataDistFBFSubsetMatrix})
also cover the i.i.d.~Gaussian setup with unknown expectation.


\section{Discussion}
\label{sec:disc}

Motivated by covariate-adjusted covariance selection under sparsity,
this paper proposes an objective Bayes method for computing the marginal likelihood
of a multivariate regression model with normally distributed errors
whose covariance matrix is constrained by a DAG.
This represents an essential ingredient to obtain a posterior probability
over the space of covariate-adjusted DAG models.
Since the proposed method is invariant with respect to Markov equivalence,
it can also be used to select covariate-adjusted decomposable models.
%
%
%
Although we do not explicitly address variable selection,
our results for the marginal likelihood can be used for Bayesian joint variable
and covariance selection, as discussed in \citet{Bhad:Mall:2013}.

In practice, as we remark at the beginning of section~\ref{sec:multivariateRegression},
variable selection is needed to apply our method
whenever the total number of predictors $p_\star$ is comparable to,
or larger than, the number of observations;
this is a typical scenario in genetical genomics applications.
Restricting our attention to models including only $p\ll n$ predictors,
so that our objective analysis becomes feasible,
turns out to be adequate for settings where
sparse models are of interest.
For instance, the two simulations
considered by \citet{Bhad:Mall:2013} have:
i) $p_\star=498$, $q=300$, and $n=120$,
with $p=11$ for the actual data generating distribution;
ii) $p_\star=498$, $q=100$, and $n=120$,
with $p=3$ for the actual data generating distribution.
Similarly, their real data analysis (eQTL Analysis on Publicly Available Human Data)
has $p_\star=3125$, $q=100$, and $n=60$, with $p=1$ or $p=2$
identified as the most likely values.



\citet{Bhad:Mall:2013} currently derive their results for decomposable models
under a weakly informative prior which requires to subjectively specify
three scalar hyper-parameters. In particular, they use a hyper-inverse Wishart
on $\bSigma$ with scale parameter equal to the identity matrix multiplied by a constant.
The latter proves to be crucial and need be fixed with 
care, because it acts as a global shrinkage parameter.
Our objective prior, with its simple method for obtaining the marginal likelihood,
should provide a useful alternative to their prior specification.
On the other hand, our methodology for computing the marginal likelihood
can also be implemented starting from a single subjectively specified matrix normal Wishart
prior under any complete DAG model, then applying the general results of
section~\ref{subsec:marginalDataDistrib}
in the context of DAG models as described in section~\ref{subsec:DAG}.
In this case, the sparsity conditions relating the sample size $n$,
the number of predictors $p$ and the maximal size of the cliques,
which we had to impose to make our objective Bayes analysis possible,
could be relaxed.

%
%
%
%
%
%
%
%

Finally, our method does not cope with non-decomposable undirected graphical models.
Bayesian covariate-adjusted covariance selection in general undirected graphical models
is beyond the scope of this paper, and will present the obvious challenge of providing
an efficient method for computing the marginal likelihood;
see \citet{Carvetal:2007}, \citet{Wang:Carv:2010}, \citet{Lenk:2013}. 
However, working within the class of decomposable graphs can still be very effective,
even when the true graph is not decomposable; see \citet{Fitchetal:2014}
for asymptotic results on the posterior model probabilities,
and for a high performing
stochastic search of the model space.

\section*{Acknowledgements}
Work partially supported by a D1-grant from Universit\`{a} Cattolica del Sacro Cuore.
The authors are grateful to Alberto Roverato for pointing out a useful reference.

\end{document}